\documentclass[11pt]{article}
\usepackage[authoryear]{natbib}
\usepackage{amssymb}
\usepackage{amsfonts}
\usepackage{amsmath}
\usepackage{rotating}
\newcommand{\cref}[2][1]{{\textup{(\hyperref[#2]{\ref*{#2}$_{#1}$})}}}
\usepackage{xr}
\usepackage{bm}
\usepackage{booktabs}
\usepackage[flushleft]{threeparttable}
\allowdisplaybreaks

\usepackage[nohead]{geometry}
\usepackage{amsthm}
\usepackage{setspace}
\usepackage{hyperref}

7

\newcommand{\var}{\mathrm{var}}
\newcommand{\beq}{\begin{eqnarray*}}
	\newcommand{\eeq}{\end{eqnarray*}}
\newtheorem{thm}{Theorem}[section]

\newtheorem{lem}{Lemma}[section]

\newtheorem{assum}{Assumption}[section]

\numberwithin{equation}{section}
\theoremstyle{definition}
\newtheorem{exm}{Example}[section]
\newtheorem{remark}{Remark}[section]
\makeatletter
\def\@biblabel#1{\hspace*{-\labelsep}}
\makeatother
\begin{document}
	
	\title{Standard Errors for Panel Data Models with Unknown Clusters}
	\date{\today }

	\author{
		Jushan Bai\thanks{%
			Address: 420 West 118th St. MC 3308, New York, NY 10027, USA. E-mail: \texttt{ jb3064@columbia.edu}.} \\ \footnotesize Columbia University
		\and Sung Hoon Choi\thanks{%
			Address: 75 Hamilton St., New Brunswick, NJ 08901, USA. E-mail:
			\texttt{sc1711@economics.rutgers.edu}.} \\  \footnotesize  Rutgers University  \and
		Yuan Liao\thanks{Address: 75 Hamilton St., New Brunswick, NJ 08901, USA. Email:
			\texttt{yuan.liao@rutgers.edu}.}\\  \footnotesize   Rutgers University
	}

	\maketitle
	
	\begin{abstract}
		This paper develops a new standard-error estimator for linear panel data models. The proposed estimator is robust to heteroskedasticity, serial correlation, and cross-sectional correlation of unknown forms.  The serial correlation is controlled by the Newey-West method. To control for cross-sectional correlations, we propose to use the thresholding method, without assuming the clusters to be known. We establish the consistency of the proposed estimator. Monte Carlo simulations show the method works well. An empirical application is considered.
	
\vspace{0.1in}	
	\end{abstract}
	
	\thispagestyle{empty}
	
	

	\onehalfspacing

	\newpage
	\setcounter{page}{1}
	\pagenumbering{arabic}
	
	\section{Introduction} \label{intro}
	Consider a linear panel regression with fixed-effects: 
	\begin{equation*}
	y_{it} = x_{it}'\beta +\alpha_{i} + \mu_{t} + u_{it}, \label{e.1}
	\end{equation*}	
	where $\alpha_{i}$ and $\mu_{t}$ are individual fixed-effects and time fixed effects; $x_{it}$ is a $k \times1$ vector of explanatory variables; $u_{it}$ is an unobservable error component. The outcome variable  $y_{it}$ and fixed effects are scalars, and $\beta$ is a $k\times 1$ vector. 
	

This paper is about the standard error of the fixed-effect ordinary least squares (OLS).
	One of the commonly used standard errors for OLS in empirical research is the \cite{white1980heteroskedasticity} heteroskedasticity robust standard error in the cross-sectional setting. In the presence of serial and cross-sectional correlations, the conventional panel standard errors may be biased. \cite{newey1986simple} introduced heteroskedasticity and autocorrelation consistent (HAC) covariance matrix estimator for time series, which allows serial correlations (also see \cite{andrews1991}, \cite{newey1994}). The cluster standard errors suggested by \cite{arellano1987} are often reported in studies of the panel model. This estimator is robust to heteroskedasticity in the cross-section and also arbitrary serial correlation, but it focuses on the large-$N$ small-$T$ scenario.  The case of large-$N$ large-$T$ is then studied by \cite{ahn2014}, \cite{hansen2007}, among many others, while either cross-sectional or serial independence is required. \cite{hansen2007} examined the covariance estimator when the time series dependence is left unrestricted.  In addition, \cite{vogelsang2012} studied the asymptotic theory that is robust to heteroskedasticity, autocorrelation, and spatial correlation, which extended and generalized the asymptotic results of \cite{hansen2007} for the conventional cluster standard errors including time fixed effects.  \cite{stock2008} suggested a bias-adjusted heteroskedasticity-robust variance matrix estimator that handles serial correlations under any sequences of $N$ or $T$.  Also, see \cite{petersen2009} who used a simulation study to examine different types of standard errors, including the clustered, Fama-MacBeth, and the modified version of Newey-West standard errors for panel data. In general, on the other hand, the conventional cluster standard errors assume that individuals across clusters are independent. Also, the cluster structure should be known such as schools, villages, industries, or states. See \cite{arellano2003}, \cite{cameron2015} and \cite{greene2003}. However, the knowledge of clusters is not available in many applications.

	In a recent interesting paper, \cite{abadie2017should} argue that clustering is an issue more of  sampling design or experimental design. Clustered standard errors are not always necessary and researchers should be more thoughtful when applying them. One reason is that clustering  may result in an unnecessarily wider confidence interval.	
Clustered standard errors are derived from the modelling perspective (model implied variance matrix) and are widely practiced, see, for example, \cite{angrist2008mostly}, \cite{cameron2005microeconometrics}, and \cite{wooldridge2003cluster, wooldridge2010econometric}.
In this paper, we continue to take the modeling perspective.
	Because of our use of thresholding method, the resulting confidence interval is not necessarily much wider, even if all cross-sectional units are allowed to be correlated.
	Furthermore, the proposed approach is also applicable when the knowledge of clustering is not available.

	We provide  a robust standard error that allows  both serial and cross-sectional correlations. We do not impose parametric structures on the serial or cross-sectional correlations. We assume these correlations are weak and apply nonparametric methods to estimate the standard errors. To control for the autocorrelation in time series, we employ the  Newey-West   truncation. To control for the cross-sectional correlation, we assume sparsity for cross-section $(i,j)$ pairs, potentially resulting from the presence of cross-sectional clusters, but the knowledge on clustering (the number of clusters and the size of each cluster) is not assumed.   We then estimate them by applying the thresholding approach of \cite{bickel2008a}.
	We also show how to make use of information on clustering when available.
	In passing we point out that instead of robust standard errors, in a separate study, \cite{baichoiliao2019} proposed a feasible GLS (FGLS) method to take into account heteroskedasticity and both serial and cross-sectional correlations. The FGLS is more efficient than OLS.
	
	The methods  we employ in this paper, banding and thresholding, are regularization methods, and have been used extensively in the recent machine learning literature for estimating high-dimensional parameters.  Nonparametric machine learning techniques have been proved to be useful tools in econometric studies.
	
	The rest of the paper is organized as follows. In Section \ref{sec2}, we describe the models and standard errors as well as the asymptotic results of OLS. Monte Carlo studies evaluating the finite sample performance of the estimators are presented in Section \ref{sec3}.  Section \ref{sec4} illustrates our methods in an application of US divorce law reform effects. Conclusions are provided in Section \ref{sec5} and all  proofs are given in Appendix \ref{app}.
	
	Throughout this paper, $\nu_{min}(A)$ and $\nu_{max}(A)$ denote the minimum and maximum eigenvalues of matrix $A$. We use $\|A\| = \sqrt{\nu_{max}(A'A)}$, $\|A\|_{1} = max_{i}\sum_{j}|A_{ij}|$ and $\|A\|_{F} = \sqrt{tr(A'A)}$ as the operator norm, the $\ell_1$-norm and the Frobenius norm of a matrix A, respectively. Note that if $A$ is a vector, $\|A\|$ is the Euclidean norm, and $|a|$ is the Absolute-value norm of a scalar $a$.
	
	\section{OLS and Standard Error Estimation} \label{sec2}
	We consider the following model:
	\begin{equation} \label{e.2.1}
	y_{it} = x_{it}'\beta + u_{it},
	\end{equation}
	where $\beta$ is a $k \times 1$ vector of unknown coefficients, $x_{it}$ is a $k \times 1$ vector of regressors, and $u_{it}$ represents the error term, often known as the idiosyncratic component. This formulation incorporates the standard fixed effects models as in \cite{hansen2007}. For example, $x_{it}, y_{it}$ and $u_{it}$ can be interpreted as variables resulting from removing the nuisance parameters from the equation, such as first-differencing to remove the fixed effects. Indeed, it is straightforward to allow additive fixed effects by using the usual demean procedure.
	
	For a fixed $t$, model (\ref{e.2.1}) can be written as:
	\begin{equation} \label{e.2.2}
	y_{t} = x_{t}\beta + u_{t},
	\end{equation}
	where $y_{t} = (y_{1t}, ..., y_{Nt})'$ $(N \times 1)$, $x_{t} = (x_{1t}, ..., x_{Nt})'$ $(N \times k),$ and $u_{t} = (u_{1t}, ..., u_{Nt})'$ $(N \times 1)$. To economize notation, we define $y_{i} = (y_{i1}, ..., y_{iT})'$ $(T \times 1)$, $x_{i} = (x_{i1}, ..., x_{iT})'$ $(T \times k),$ and $u_{i} = (u_{i1}, ..., u_{iT})'$ $(T \times 1)$. So when the vector $y$ is indexed by $t$, it refers to an $N\times 1$ vector, and when $y$ is indexed by $i$ it refers to a $T\times 1$ vector. Similar meaning is applied to $x$ and $u$. There is no confusion when context is clear.

	The (pooled) ordinary least square (OLS) estimator of $\beta$ from equations (\ref{e.2.1}) and (\ref{e.2.2}) may   be defined as
	\begin{equation} \label{e.2.3}
	\widehat{\beta} = (\sum\limits_{i=1}^{N}\sum\limits_{t=1}^{T}x_{it}x_{it}')^{-1}\sum\limits_{i=1}^{N}\sum\limits_{t=1}^{T}x_{it}y_{it} = (\sum\limits_{t=1}^{T}x_{t}'x_{t})^{-1}\sum\limits_{t=1}^{T}x_{t}'y_{t}.
	\end{equation}
	
	The variance of $\widehat{\beta}$ depends on both $V_{X} \equiv \cfrac{1}{NT}\sum\limits_{i=1}^{N}\sum\limits_{t=1}^{T}x_{it}x_{it}'$, and particularly,  
	\begin{align} 
	V & \equiv  Var(\cfrac{1}{\sqrt{NT}}\sum\limits_{i=1}^{N}\sum\limits_{t=1}^{T}x_{it}u_{it})  \cr
	& =  \cfrac{1}{NT}\sum\limits_{t=1}^{T}Ex_{t}'u_{t}u_{t}'x_{t} +  \cfrac{1}{NT}\sum\limits_{h=1}^{T-1}\sum\limits_{t=h+1}^{T}[Ex_{t}'u_{t}u_{t-h}'x_{t-h}+Ex_{t-h}'u_{t-h}u_{t}'x_{t}].\label{e.2.4}
	\end{align}
	The goal of this paper is to consistently estimate $V$ in the presence of both serial and cross-sectional correlations in $\{u_{it}\}$.

	There are two types of clustered standard errors suggested by \cite{arellano1987}. The original individual clustered version is
	\begin{equation*} \label{e.9}
	\widehat{V}_{CX} = \frac{1}{NT}\sum_{i=1}^{N}x_{i}'\widehat{u}_{i}\widehat{u}_{i}'x_{i},
	\end{equation*}
	with $\widehat{u}_{i} = y_{i}-x_{i}\widehat{\beta} $ are the OLS residuals, and this estimator allows for arbitrary serial dependence and heteroskedasticity within individuals. In addition, $\widehat{V}_{CX}$ assumes no cross-section correlation.

	The time-clustered version, which allows for heteroskedasticity and arbitrary cross-sectional correlation, is
	\begin{equation*} \label{e.10}
	\widehat{V}_{CT} = \frac{1}{NT}\sum_{t=1}^{T}x_{t}'\widehat{u}_{t}\widehat{u}_{t}'x_{t},
	\end{equation*}
	with $\widehat{u}_{t} = y_{t}-x_{t}\widehat{\beta}$. Here $\widehat{V}_{CT}$ assumes no serial correlation. 
	
	The above clustered standard errors are robust to either  arbitrary serial correlation or arbitrary cross-sectional correlation, respectively. In practice, however, since the dependence assumption is unknown, an over-rejection problem may occur. Specifically, if there exist both serial and cross-sectional correlations, these estimators are not robust anymore, as our numerical evidence shows in Section \ref{sec3} (e.g., Tables \ref{bothcorr3} and \ref{error factor structure}).
	
	To control for the serial correlation, a simple modification of $\widehat{V}_{CT}$ using \cite{newey1986simple} is
	\begin{equation} \label{e.2.5}
	\widehat{V}_{DK} = \cfrac{1}{NT}\sum\limits_{t=1}^{T}x_{t}'\widehat{u}_{t}\widehat{u}_{t}'x_{t} +  \cfrac{1}{NT}\sum\limits_{h=1}^{L}\omega(h,L)\sum\limits_{t=h+1}^{T}[x_{t}'\widehat{u}_{t}\widehat{u}_{t-h}'x_{t-h}+x_{t-h}'\widehat{u}_{t-h}\widehat{u}_{t}'x_{t}],
	\end{equation}
	 where $\omega(\cdot)$ is the kernel function and $L$ is the bandwidth. This estimator is suggested by \cite{driscoll1998}. When $N$ is large, however, (\ref{e.2.5}) accumulates a large number of  cross-sectional estimation noises.

	More generally, let
	\begin{equation*} \label{e.12}
	V_{ij} \equiv  \cfrac{1}{T}\sum\limits_{t=1}^{T}Ex_{it}u_{it}u_{jt}x_{jt}' +  \cfrac{1}{T}\sum\limits_{h=1}^{T-1}\sum\limits_{t=h+1}^{T}[Ex_{it}u_{it}u_{j,t-h}x_{j,t-h}'+Ex_{i,t-h}u_{i,t-h}u_{jt}x_{jt}'].
	\end{equation*}
	Then  equation (\ref{e.2.4}) can be written as
	\begin{equation*}
	V = \cfrac{1}{N}\sum\limits_{ij}V_{ij}.
	\end{equation*}
	Unlike time series observations, cross-sectional observations have no natural ordering. They can be arranged in different orders. That is why cross-sectional correlation is more difficult to control.
	The usual cluster standard error makes the following assumption: let $C_1,...,C_G$  be disjoint subsets of $\{1,...,N\}$, so that they are \textit{known} clusters and that $V_{ij}=0$ when $i$ and $j$ belong to different clusters.  So $V$ can be expressed as
	$$
	V= \frac{1}{N}\sum_{g=1}^G\sum_{(i,j)\in C_g} V_{ij}.
	$$
	See \cite{liang1986}. Suppose the cardinality of each $C_g$ is small (this would be the case if the number of clusters $G$ is large) or grows slowly with $N$,  then we only need to estimate $\sum_{g=1}^G\sum_{(i,j)\in C_g}1$ number of $V_{ij}$'s, greatly reducing the number of pair-wise covariances. But as commented in the literature, this requires the knowledge of $C_1,...,C_G$, which in some applications, is not naturally available.

	\subsection{The estimator of $V$ with unknown clusters} \label{estimator}
	
	The key assumption  we make is that conditionally on $x_{t}, \{u_{it}\}$ is weakly correlated across both $t$ and $i$. Essentially, this means $ V_{ij}$ is zero or nearly so for most  pairs of $(i, j)$.  There is a partition $\{(i,j) : i,j \leq N\} = S_{s} \bigcup S_{l}$ so that
	\begin{align*} \label{e.14}
	S_{s} = \{(i, j) : \|Ex_{it}u_{it}u_{j,t+h}x_{j,t+h}'\| = 0 \, \forall h\},\\~
	S_{l} = \{(i, j) : \|Ex_{it}u_{it}u_{j,t+h}x_{j,t+h}'\| \neq 0 \, \exists h\},
	\end{align*}
	where the subscript ``$s$" indicates ``small", and ``$l$" indicates ``large".
	We assume that $(i,i) \in S_l$ for all $i\leq N$, and importantly,  most pairs $(i,j)$ belong to $S_s$.
	Yet, we do not need to know which elements belong to $S_{s}$ or $S_{l}$. Then
	\begin{equation*}
	V  = \frac{1}{N}\sum_{(i,j) \in S_l}V_{ij}.
	\end{equation*}
Furthermore, let $\omega(h, L) = 1-h/(L+1)$  be the Bartlett kernel.
	Also see \cite{andrews1991} for other kernel functions.
	As suggested by \cite{newey1986simple}, $V_{ij}$ can be approximated by
	\begin{equation*} \label{e.13}
	V_{u,ij} \equiv \cfrac{1}{T}\sum\limits_{t=1}^{T}Ex_{it}u_{it}u_{jt}x_{jt}'+\cfrac{1}{T}\sum\limits_{h=1}^{L}\omega(h,L)\sum\limits_{t=h+1}^{T}[Ex_{it}u_{it}u_{j,t-h}x_{j,t-h}' + Ex_{i,t-h}u_{i,t-h}u_{jt}x_{jt}'].
	\end{equation*}
	Then approximately,
	\begin{equation*} \label{e.15}
	V  \approx \frac{1}{N}\sum_{(i,j) \in S_l}V_{u,ij}.
	\end{equation*}
	The above approximation plays the fundamental role of our standard error estimator.   We estimate $V_{ij}$ using \cite{newey1986simple}, and estimate $S_l$ using the  cross-sectional thresholding.

 To apply \cite{newey1986simple},   we estimate $V_{u,ij}$ by
	\begin{equation*} \label{e.16}
	S_{u,ij} \equiv \cfrac{1}{T}\sum\limits_{t=1}^{T}x_{it}\widehat{u}_{it}\widehat{u}_{jt}x_{jt}'+\cfrac{1}{T}\sum\limits_{h=1}^{L}\omega(h,L)\sum\limits_{t=h+1}^{T}[x_{it}\widehat{u}_{it}\widehat{u}_{j,t-h}x_{j,t-h}' + x_{i,t-h}\widehat{u}_{i,t-h}\widehat{u}_{jt}x_{jt}'],
	\end{equation*}
	where $\widehat{u}_{it} = y_{it}-x_{it}'\widehat{\beta}$.
	For a predetermined threshold value $\lambda_{ij}$,
	we  approximate $S_l$ by
	\begin{equation*} \label{e.17}
	\widehat{S}_l = \{(i,j) : \|S_{u,ij}\| > \lambda_{ij}\}.
	\end{equation*}
	Hence, a ``matrix hard-thresholding" estimator of $V$ is
	\begin{equation*} \label{e.18}
	\widehat{V}_{\text{Hard}} \equiv \frac{1}{N}\sum_{(i,j) \in \widehat{S}_l\cup \{i=j\}} S_{u,ij}.
	\end{equation*} As for the threshold value, we specify
	\begin{equation*} \label{e.19}
	\lambda_{ij} = M \, \omega_{NT}\sqrt{\|S_{u,ii}\|\|S_{u,jj}\|}, \text{ where }  \omega_{NT} = L\sqrt{\frac{\log(LN)}{T}}
	\end{equation*}
	for a constant $M>0$.
	The converging sequence $\omega_{NT} \rightarrow 0$ is chosen to satisfy:
	\begin{equation*} \label{e.20}
	\max\limits_{i,j\leq N}\|S_{u,ij}-V_{u,ij}\| = O_{P}(\omega_{NT}).
	\end{equation*}
	In practice, the thresholding constant, $M$, can be chosen through multifold cross-validation, which is discussed in the next subsection. In addition, we can obtain $\widehat{V}_{DK}$ from $\widehat{V}_{\text{Hard}}$ by setting $M=0$.

	We also recommend a ``matrix soft-thresholding" estimator as follow:
	\begin{equation*} \label{e.21}
	\widehat{V}_{\text{Soft}} \equiv \frac{1}{N}\sum_{i,j} \widehat{S}_{u,ij},
	\end{equation*}
	where $\widehat{S}_{u,ij}$ is
	\begin{equation*} \label{e.22}
	\widehat{S}_{u,ij}=\begin{cases}S_{u,ij},&\text{if}\,\,i=j,\\
	A_{u,ij},&\text{if}\, \|S_{u,ij}\|>\lambda_{ij},  \text{ and } i\neq j,\\
	0,&\text{if}\, \|S_{u,ij}\|<\lambda_{ij},  \text{ and } i\neq j,
	\end{cases}
	\end{equation*}
	where the $(k,k')$'s element of $A_{u,ij}$ is ($sgn(x)$ denotes the sign function)
	\begin{equation*} \label{e.23}
	A_{u,ij,kk'}=\begin{cases}
	sgn(S_{u,ij,kk'})[|S_{u,ij,kk'}|-\eta_{ij,kk'}]_{+},&\text{if}\, |S_{u,ij,kk'}|>\eta_{ij,kk'}, \\
	0,&\text{if}\,|S_{u,ij,kk'}|<\eta_{ij,kk'},
	\end{cases}
	\end{equation*}
	for the threshold value
	\begin{equation*} \label{e.24}
	\eta_{ij,kk'} = M \, \omega_{NT}\sqrt{|S_{u,ii,kk'}||S_{u,jj,kk'}|}, \text{ where }  \omega_{NT} = L\sqrt{\frac{\log(LN)}{T}}
	\end{equation*}
	for some constant $M>0$.\\

	\begin{remark}
		The thresholding estimators for $V$ do not assume known cluster information (the number of clusters and the membership of clusters). The method can also be modified to take into account the clustering information when available,  and is particularly suitable when the number of clusters is small, and the size of each cluster is large. The modification is to apply the thresholding method within each cluster. The conventional clustered standard errors lose a lot of degrees of freedom when the size of cluster is too large (because each cluster is effectively treated as a ``single observation"), resulting in conservative confidence intervals.  See \cite{cameron2015}.   The thresholding  avoids this problem, while allowing correlations of unknown form within each cluster.
	\end{remark}
	
	\subsection{Choice of tuning parameters} \label{tunning}
	Our suggested estimators, $\widehat{V}_{\text{Hard}}$ and $\widehat{V}_{\text{Soft}}$, require the choice of tuning parameters $L$ and $M$, which are the bandwidth and the threshold constant respectively. To choose the bandwidth $L$, we recommend using $L = 4(T/100)^{2/9}$ as \cite{newey1994} suggested.
	
	In practice, $M$ can be chosen through multifold cross-validation. After obtaining the estimated residuals $\widehat{u}_{it}$ by OLS, we split the data into two subsets, denoted by $\{\widehat{u}_{it}\}_{t \in J_{1}}$ and $\{\widehat{u}_{it}\}_{t \in J_{2}}$;  let $T(J_{1})$ and $T(J_{2})$ be the sizes of $J_{1}$ and $J_{2}$, which are $T(J_{1})+ T(J_{2}) = T$ and $T(J_{1})\asymp T$. As suggested by \cite{bickel2008a}, we can set $T(J_{1})=T (1-\log(T)^{-1})$ and $T(J_{2})=T/\log(T)$; $J_1$ represents the training data set, and $J_2$ represents the validation data set.
	
	The procedure requires splitting the data multiple times, say $P$ times. At the $p$th split, we denote by $\widehat{V}^{p}$ the sample covariance matrix based on the validation set, defined by
\[ \widehat{V}^{p} = \frac{1}{N}\sum_{ij}S_{u,ij}^{p}, \]
with  $S_{u,ij}^{p}$ defined similarly to $S_{u,ij}$  using data on $J_2$.
 Let $\widehat{V}_{s}(M)$ be the thresholding estimator with threshold constant $M$ using the entire sample.
Then we choose the constant $M^{*}$ by minimizing a cross-validation objective function
	
	\begin{equation*}
	M^* = \arg \min_{0<M<M_0}\frac{1}{P}\sum_{p=1}^{P}\|\widehat{V}_{s}(M)-\widehat{V}^{p}\|_{F}^2,\quad s\in\{\text{Hard}, \text{Soft}\}
	\end{equation*}
	and the resulting estimator is $\widehat{V}_{s}(M^*)$. We use $L=4(T/100)^{2/9}$ for both $\widehat{V}_{s}(M)$ and $\widehat{V}^{p}$ and  find that setting $M_0=1$ works well.
	So the minimization is taken over $M\in (0,1)$ through a grid search.

The above procedure modifies that of \cite{bickel2008a}   in two aspects. One is to use the entire sample when computing $\widehat{V}_{s}$  instead of $J_1$. Since $T(J_{1})$ is close to $T$, this modification does not change the result much, but simplifies the computation.  The second modification is to use a consecutive block for the validation set because of time series, so that the serial correlation is not perturbed. Hence in view of the time series nature, we first divide the data into $P=\log(T)$ blocks with block length $T/\log(T)$.  Each $ J_{2} $ is taken as one of the $P$ blocks when computing $\widehat V^p$, similar to the K-fold cross validation. 
We have   conducted simulations of the cross-validation in the presence of both correlations, and the results show that this procedure performs well. For instance, the cross-validation tends to choose smaller $M$ as the cross-sectional correlation becomes stronger. Due to the page limit, however, those are not reported in this paper.

	\subsection{Consistency}
	Below we present assumptions under which $\widehat{V}$ (either $\widehat{V}_{\text{Hard}}$ or $ \widehat{V}_{\text{Soft}}$) consistently estimates $V$. We define
	\begin{equation*} \label{e.25}
	\alpha_{NT}(h) \equiv \sup\limits_{X}\max\limits_{t\leq T}[\|E(u_{t}u_{t-h}'|X)\|+\|E(u_{t-h}u_{t}'|X)\|]
	\end{equation*}
	and
	\begin{equation*} \label{e.26}
	\rho_{ij,h} \equiv \sup\limits_{X}\max\limits_{t\leq T}[|E(u_{it}u_{j,t-h}|X)|+|E(u_{i,t-h}u_{jt}|X)|],
	\end{equation*}
    where $X=\{x_{it}\}_{i\leq N,t\leq T}$. These coefficients give measures of autocovariances and cross-section covariances.

	\begin{assum}	\label{assumption1}
		(i) $E(u_{t}|x_{t})=0$.\\
		(ii) Let $\nu_1 \leq ... \leq \nu_k$ be the eigenvalues of $(\frac{1}{NT}\sum_{i=1}^{N}\sum_{t=1}^{T}Ex_{it}x_{it}')$. Then there exist constants $c_{1}, c_{2} > 0$ such that $c_1 < \nu_1\leq \cdots \leq \nu_k < c_2$.\\
	\end{assum}
	
	\begin{assum} \label{assumption2}
		(weak serial and cross-sectional dependence). \\
		(i) $\sum_{h=0}^{\infty}\alpha_{NT}(h) \leq C$ for some $C >0$. In addition,
	there exist $\kappa\in(0,1)$,  $C>0$ such that for all $T>0$, 		$$
		\sup\limits_{A\in \mathcal{F}_{-\infty}^0, B \in \mathcal{F}_{T}^{\infty}}|P(A)P(B)-P(AB)| < exp(-CT^{\kappa}),
		$$
		where $\mathcal{F}_{-\infty}^0$ and $\mathcal{F}_{T}^{\infty}$ denote the $\sigma$-algebras generated by $\{(x_{t},u_{t}) : t \leq 0\}$ and $\{(x_{t},u_{t}) : t \geq T\}$ respectively.

		(ii) For some $q \in [0,1),$ $\omega_{NT}^{1-q}\max_{i\leq N}\sum_{j=1}^{N}(\sum_{h=0}^{L}\rho_{ij,h})^q = o(1),$ where  $\omega_{NT} \equiv L \sqrt{\frac{\log(LN)}{T}}$. 
		
	\end{assum}
	
	Assumption \ref{assumption2} (i) is the standard alpha-mixing condition, adapted to the large-$N$ panel. 
			 Condition (ii) is new here. It requires weak cross-sectional correlations. It is similar to the ``approximate sparse assumption" in \cite{bickel2008a}. Note that we actually allow the presence of many ``small" but nonzero $\|Ex_{it}u_{it}u_{j,t+h}x_{j,t+h}'\|$. Clusters that have ``large" $\|Ex_{it}u_{it}u_{j,t+h}x_{j,t+h}'\|$ are unknown to us. Hence the appealing feature of our method is that we allow for unknown clusters. 
			 
			 Essentially the assumption $\omega_{N,T}^{1-q}\max_{i\leq N}\sum_{j \leq N}(\sum_{h=0}^{L}\rho_{ij,h})^{q} = o(1)$ controls the order of elements in $S_l$.
		The following example presents a case of cross sectional weak correlations that satisfies condition (ii). 
	\begin{exm}
	Suppose uniformly for all $h=0,..., L$,  $E(u_{t}u_{t-h}'|X)$ is an $N\times N$ block-diagonal matrix, where the  size of each block is at most $S_{NT}$, which practically means  that  each cluster contains no more than $S_{NT}$  individuals, assuming clusters are mutually uncorrelated.  Then $\rho_{ij,h}=0$ for $(i,j)$ belong to different blocks. Within the same block,  almost surely in $X$, 	$$
 |E(u_{i,t}u_{j,t-h}|X)| +	 |E(u_{i,t-h}u_{jt}|X)| \leq \alpha_{NT}(h),\quad \sum_{h=0}^{\infty}\alpha_{NT}(h)<\infty
	$$
	Then let $B(i)$ denote the block that $i$ belongs to, whose size is at most $S_{NT}$.
	\begin{eqnarray*}
	\omega_{NT}^{1-q}\max_{i\leq N}\sum_{j=1}^{N}(\sum_{h=0}^{L}\rho_{ij,h})^q &=&
	\omega_{NT}^{1-q}\max_{i\leq N}\sum_{j\in B(i)}(\sum_{h=0}^{L}\rho_{ij,h})^q \cr
	&\leq& C\omega_{NT}^{1-q} S_{NT}(\sum_{h=0}^{\infty}\alpha_{NT}(h))^{q/c}
	\end{eqnarray*}
for constants $c,C>0.$
	The last term converges to zero so long as $\omega_{NT}^{1-q} S_{NT}\to 0$. This then requires either fixed or slowly growing cluster size $S_{NT}$.  

	\end{exm}
	
	\begin{assum} \label{assumption3}
	
	(i) For each fixed $h$, $\omega(h, L) \rightarrow 1$ as $L \rightarrow \infty$ and $\max_{h\leq L} |\omega(h,L)| \leq C$ for some $C > 0.$
	\\
	(ii)   Exponential tail: There exist $r_{1}, r_{2}>0$ and $b_{1}, b_{2} > 0$, such that $ r_{1}^{-1}+r_{2}^{-1}+\kappa^{-1}>1$, and for any $s > 0, i\leq N$,  
		$$P(|u_{it}| > s) \leq exp(-(s/b_{1})^{r_1}),\quad P(|x_{it}|>s) \leq \exp(-(s/b_2)^{r_2}).$$
			(iii) There is $c_{1} > 0, $ for all $i,  \lambda_{\min}( \var(\frac{1}{\sqrt{T}}\sum_{t=1}^Tx_{it}u_{it})) > c_{1}.$  Additionally, 	the eigenvalues of $V$ and $V_{X}$ are bounded away from both zero and infinity.
 
	\end{assum}

	Condition (i) is well satisfied by various kernels for the HAC-type estimator.  Condition (ii)  ensures the Bernstein-type inequality for weakly dependent data.  Note that it requires the underlying distributions to be thin-tailed. Allowing for heavy-tailed distributions is also an important issue. However, it would require a very different estimation method, and is out of the scope of this paper. 
	Nevertheless, we have conducted simulation studies under heavy-tailed distributions (e.g., $t$-distribution with degree of freedom 5). Indeed, the proposed estimator works well in this case, even though the theory requires thin-tailed distributions.\footnote{The simulation results for the heavy-tailed distributions are available upon request from the authors.}

	We have the following main theorem and all proofs are contained in Appendix A1.\\~
	
	\begin{thm}\label{asdistribution}
	 Under Assumption \ref{assumption1}-\ref{assumption3}, as $N,T \rightarrow \infty$,
		\begin{equation*}
		\sqrt{NT}[V_{X}^{-1}\widehat{V}V_{X}^{-1}]^{-1/2}(\widehat{\beta}-\beta)\overset{d}{\to} \mathcal{N}(0,I).
		\end{equation*}
	\end{thm}
	Theorem \ref{asdistribution} allows us to construct a $(1-\tau)\%$ confidence interval for $c'\beta$ for any given $c \in \mathbb{R}^{k}$. The standard error of $c'\hat \beta_{OLS}$ is
\[  \Big( \frac 1 {NT} c'(V_{X}^{-1}\widehat{V}V_{X}^{-1})c  \Big)^{1/2} \]
and the confidence interval for $c'\beta$ is
  $[c'\widehat{\beta} \pm Z_{\tau}\hat{\sigma}/\sqrt{NT}]$ where $Z_{\tau}$ is the $(1-\tau)\%$ quantile of standard normal distribution and $\hat{\sigma} =  (c'(V_{X}^{-1}\widehat{V}V_{X}^{-1})c)^{1/2}$.

	\section{Monte Carlo Experiments} \label{sec3}
	\subsection{DGP and methods} \label{dgp}
	In this section we examine the finite sample performance of the robust standard errors using simulation study.    The data generating process (DGP) used for the simulation is produced by the fixed effect linear regression model
	\begin{equation*}
	y_{it} =  \alpha_{i} + \mu_{t}+\beta_{0}x_{it} + u_{it},
	\end{equation*}
	where the true $\beta_{0} =1$. The DGP allows for serial and cross-sectional correlations in $x_{it}$ as follow:
	\begin{align*}
	&x_{it} = a_{i}\nu_{i+1,t}+ \nu_{i,t}+b_{i}\nu_{i-1,t}, \;\; \nu_{it}= \rho_{X} \nu_{i,t-1}+\epsilon_{it},\;\;
	\epsilon_{it} \sim N(0,1),\;\; \nu_{i0}=0,\\
	&\alpha_{i} \sim N(0,0.5), \;\; \mu_{t} \sim N(0,0.5),
	\end{align*}
	where the constants $\{a_{i}, b_{i}\}_{i=1}^{N}$ are i.i.d. Uniform$(0,\gamma_{X})$, which introduce cross-sectional correlation. In addition, $\nu_{it}$ is modeled as $AR(1)$ process with the autoregressive parameter $\rho_{X}$. Throughout this simulation study, we set $\rho_X = 0.3$ and $\gamma_{X}=1$.
	
	We generate the error terms, $u_{it}$, in three different cases as follow:
	\begin{align*}
	\text{Case 1:}\;\; & u_{it} = c_{i}m_{i+1,t}+ m_{i,t}+d_{i}m_{i-1,t}, \;\; m_{it}= \rho m_{i,t-1}+\varepsilon_{it},\;\; \varepsilon_{it} \sim N(0,1),\;\; m_{i0}=0,\\
	\text{Case 2:}\;\; & u_{it} = \psi\sum_{j=1}^{N}w_{ij}u_{it}+\eta_{it}, \;\; \eta_{it} \sim N(0,1), \;\; u_{i0} =0, \\
	\text{Case 3:}\;\; & u_{it} = \sum_{k=1}^{r}\lambda_{ij}F_{tk}+e_{it}, \;\; F_{tk}= \rho_{F}F_{t-1,k}+\xi_{tk}, \;\; \lambda_{ik}= \rho_{\lambda}\lambda_{i-1,k}+\zeta_{tj},\;\;\\
	&e_{it} \sim N(0,1), \;\; \xi_{it} \sim N(0,1), \;\; \zeta_{it} \sim N(0,1).	
	\end{align*}
	The regressor is uncorrelated with the error term $u_{it}$ each other. In Case 1, we generate the error term similar to $x_{it}$. The constants $\{c_{i}, d_{i}\}_{i=1}^{N}$ are  i.i.d. Uniform$(0,\gamma)$, which introduce cross-sectional correlation, and heteroskedasticity when  $\gamma>0$. $m_{it}$ is modeled as $AR(1)$ process with the autoregressive parameter $\rho$. Varying $\gamma>0$ allows us to control for the strength of the cross-sectional correlation. Data are generated with four different structures of regressors and  error terms: (a) no correlations  ($\rho = 0, \gamma = 0$); (b) only serial correlation ($\rho = 0.5, \gamma = 0$); (c) only cross-sectional correlation ($\rho = 0, \gamma = 1$); and (d) both serial and cross-sectional correlations ($\rho = \{0.3,0.9\}, \gamma = 1$). In Case 2, the error terms are modeled as a spatial autoregressive (SAR(1)) process. The matrix $W=(w_{ij})_{N\times N}$ is a rook type weight matrix whose diagonal elements are zero. Note that the rows of $W$ are standardized, hence they sum to one. $\psi$ is the scalar spatial autoregressive coefficient with $|\psi|<1$. In this paper, we report the case of $\psi=0.5$. Importantly, SAR(1) model does not produce the serial correlation on the error term.  In Case 3, we consider an error factor structure. Both factors and factor loadings follow AR(1) processes, which introduce both serial and cross-sectional correlations. We set $r=2$, and consider the cases of $\rho_{\lambda} = 0.3$ and $\rho_{F} = 0.9$. 
	
	In this simulation study, we examined $t$-statistics for testing the null hypothesis $H_{0}: \beta_{0}=1$ against the alternative $H_{1}:\beta_{0} \neq 1$. In each simulation we compare the proposed estimator with that of other common five types of standard errors for $\widehat{\beta}$: the standard White estimator given by $\widehat{V}_{White} = \frac{1}{NT}\sum_{i=1}^{N}\sum_{t=1}^{T}\widetilde{x}_{it}\widetilde{x}_{it}'\widehat{u}_{it}^2$, where $\tilde{x}_{it}$ is demeaned version of regressor. Two types of clustered standard errors, $\widehat{V}_{CX}$ and $\widehat{V}_{CT}$, as defined in Section 2. In addition, we use two types of Newey and West HAC estimators for the panel version as follows: \\
	$$\widehat{V}_{DK} = \frac{1}{NT}\sum_{t=1}^{T}\widetilde{x}_{t}'\widehat{u}_{t}\widehat{u}_{t}'\widetilde{x}_{t} +  \frac{1}{NT}\sum_{h=1}^{L}\omega(h,L)\sum_{t=h+1}^{T}[\widetilde{x}_{t}'\widehat{u}_{t}\widehat{u}_{t-h}'\widetilde{x}_{t-h}+\widetilde{x}_{t-h}'\widehat{u}_{t-h}\widehat{u}_{t}'\widetilde{x}_{t}]$$ and
	$$\widehat{V}_{HAC} = \frac{1}{NT}\sum_{i=1}^{N}\sum_{t=1}^{T}\widetilde{x}_{it}\widetilde{x}_{it}'\widehat{u}_{it}^2 + \frac{1}{NT}\sum_{i=1}^{N}\sum_{h=1}^{L}\omega(h,L)\sum_{t=h+1}^{T}[\widetilde{x}_{it}\widehat{u}_{it}\widehat{u}_{i,t-h}\widetilde{x}_{i,t-h}'+\widetilde{x}_{i,t-h}\widehat{u}_{i,t-h}\widehat{u}_{it}\widetilde{x}_{it}']. $$
	Note that $\widehat{V}_{HAC}$ assumes cross-sectional independence, while $\widehat{V}_{DK}$ allows arbitrary cross-sectional dependence. In addition, $\widehat{V}_{DK}$ and $\widehat{V}_{HAC}$ can be obtained from our proposed estimator with $M=0$ and a large constant $M$, respectively.
	
	Results are given for sample sizes $N = 50, 200$ and $T = 100, 200$. For each $\{N, T\}$ combination, we set $L=3,7,11$ as the bandwidth for $\widehat{V}_{HAC}$, $\widehat{V}_{DK}$, and the proposed estimator, $\widehat{V}_{\text{Hard}}$. We also use Bartlett kernel for these three estimators. For the thresholding constant parameters of $\widehat{V}_{\text{Hard}}$, we set $M = 0.10, 0.15, 0.20, 0.25$ in all cases. The simulation is replicated for one thousand times for each case and the nominal significance level is 0.05. Simulation results are reported in Tables \ref{no_cross} - \ref{error factor structure}.

	\subsection{Results}
	Tables \ref{no_cross} - \ref{error factor structure} present the simulation results, where each table corresponds to different cases. Each table presents results of null rejection probabilities for 5\% level tests based on six different standard errors. As expected, a common feature in all tables is that when both $N$ and $T$ are small, all six estimators have rejection probabilities greater than 0.05. This might happen even when the errors are drawn from i.i.d. standard normal, and this problem becomes more noticeable in the presence of serial, cross-sectional, or both correlations. A number of interesting findings based on  tables are summarized below.
	
	Tables \ref{no_cross}-\ref{bothcorr3} shows the results of Case 1. In Table \ref{no_cross}, Panel A indicates that all the estimators perform well due to no correlation. Especially, White standard error estimators give rejection probabilities close to 0.05. In Panel B, when the serial correlation is introduced, the performances of $\widehat{V}_{CX}$ and $\widehat{V}_{HAC}$ are markedly better than others except for small sample size. In addition, our proposed estimator, $\widehat{V}_{\text{Hard}}$, also performs well when we use both larger threshold constant $M$ and bandwidth $L$. Since there is only a serial correlation in the error term, these estimators take this correlation into account and perform well. As the size of bandwidth increases, the standard error estimated by $\widehat{V}_{HAC}$ increases to a level similar to the results of $\widehat{V}_{CX}$ and the tendency to over-reject diminishes. Since the Newey-West technique gives the weight, which is less than one, the estimated standard error may be underestimated. Hence, the traditional cluster standard error, $\widehat{V}_{CX}$, dominates the standard error of Newey-West panel version, $\widehat{V}_{HAC}$. Note that the unreported rejection probabilities of $\widehat{V}_{DK}$ exponentially increases as the bandwidth $L$ increases.
	
	Table \ref{crosssectional} considers the case of cross-sectionally correlated errors and regressors. In Panel A, except the case of small sample size, $\widehat{V}_{CT}$ and $\widehat{V}_{DK}$ with small bandwidth $L$ have rejection probabilities close to 0.05 in the first panel. Also, $\widehat{V}_{\text{Hard}}$ with small $L$ and $M$ performs well. Importantly, notice that the rejection rate of $\widehat{V}_{DK}$ and $\widehat{V}_{\text{Hard}}$ tend to over-reject substantially as the lag length $L$ increases. In addition, as the cross-section size $N$ increases, the over-rejection problem becomes worse, as we mentioned in Section \ref{sec2}. This tendency is easy to explain. Since $\widehat{V}_{DK}$ is an estimator based on a single time series and it is zero when full weight is given to the sample autocovariance, the bias in $\widehat{V}_{DK}$ initially falls but then increases as the lag length increases, while the variance of $\widehat{V}_{DK}$ is initially increasing but eventually becomes decreasing. Hence, $\widehat{V}_{DK}$ is biased downward substantially, and its t-statistics tends to over-reject when a large bandwidth is used. On the other hand, in the case of the small size of $L$ and $M$, $\widehat{V}_{\text{Hard}}$ gives less bias on the estimated standard error.
	
	Panel B of Table \ref{crosssectional} allows the serial correlation as well as the cross-sectional correlation. Not surprisingly, all estimators except $\widehat{V}_{\text{Hard}}$ and $\widehat{V}_{DK}$ tend to over-reject substantially. In the small sample, these two estimators get worse than the case of the first panel. In the large sample, however, rejection probabilities of $\widehat{V}_{\text{Hard}}$ and $\widehat{V}_{DK}$ are close to 0.05. Importantly, $\widehat{V}_{\text{Hard}}$ outperforms $\widehat{V}_{DK}$ by choosing $M$ properly. Unreported results of $\widehat{V}_{DK}$ with larger bandwidth, $L$, show much larger rejection probabilities than that of $\widehat{V}_{\text{Hard}}$.  This indicates that we can obtain unbiased standard error estimator and appropriate rejection rates using our proposed estimators,  $\widehat{V}_{\text{Hard}}$. 
	Table \ref{bothcorr3} is the result of strong serial correlation with the cross-sectional dependence. When the serial correlation gets stronger, such as $\rho=0.9$, all estimators tend to over-reject exponentially in small samples. However, $\widehat{V}_{\text{Hard}}$ and $\widehat{V}_{DK}$ outperform other estimators as the dimensionality increases. 
	
	Table \ref{SAR} considers the error with SAR(1) structure, which does not require the serial correlation on the error term. Similar to the results reported in the first panel of Table \ref{crosssectional}, $\widehat{V}_{CT}$ gives rejection probabilities close to 0.05. $\widehat{V}_{DK}$ and $\widehat{V}_{\text{Hard}}$ with small bandwidth $L$ also perform well. Moreover, $\widehat{V}_{\text{Hard}}$ with proper thresholding constant $M$ gives less bias than $\widehat{V}_{DK}$ on the estimated standard error.

	Finally, Table \ref{error factor structure} presents the results of the error factor structure.  Similar to the results of Table \ref{bothcorr3}, all estimators except $\widehat{V}_{\text{Hard}}$ and $\widehat{V}_{DK}$ tend to over-reject. Rejection probabilities of $\widehat{V}_{\text{Hard}}$ and $\widehat{V}_{DK}$ are relatively close to 0.05 when the sample size is large. 

	\section{Empirical study: Effects of divorce law reforms} \label{sec4}
	In this section, we re-examine the empirical work of the association between divorce law reforms and divorce rates using our proposed OLS standard error. There are many empirical studies on the effects of divorce law reforms on divorce rates. \cite{friedberg1998} found that state law reforms significantly increased divorce rates with controls for state and year fixed effects. 
	\cite{wolfers2006did} investigated the question of whether law reform continues to have an impact on the divorce rate by including dummy variables for the first two years after the reforms, 3-4 years, 5-6 years, and so on. Specifically, he studied the following fixed effect panel data model
	\begin{equation}
	y_{it} =  \alpha_{i} + \mu_{t} +  \sum_{k=1}^{8}\beta_{k}X_{it,k} + \delta_{i}t + u_{it},
	\end{equation}
	where $y_{it}$ is the divorce rate for state $i$ and year $t$; $\alpha_{i}$ and $\mu_{t}$ are the state and year fixed effects; $X_{it,k}$ is a binary regressor that representing the treatment effect $2k$ years after the reform; $\delta_{i}t$ a linear time trend. 
	\cite{wolfers2006did} suggested that there might be two sides of the same treatment yield this phenomenon: a number of divorces gradually shifted after the earlier dissolution of bad matches, after the reform.
	
	Both \cite{friedberg1998} and \cite{wolfers2006did} estimated OLS regressions using state population weight for each year. In addition, they estimated standard errors under the assumption that errors are homoskedastic, serially and cross-sectionally uncorrelated. However, ignoring these correlations might lead to bias in the standard error estimators. We re-estimated the model of \cite{wolfers2006did} using proposed OLS standard error estimators.
	
	The same data as in \cite{wolfers2006did} are used, but we exclude Indiana, New Mexico and Louisiana due to missing observations around divorce law reforms. As a result, we obtain a balanced panel data contain the divorce rates, state-level reform years and binary regressors from 1956 to 1988 over 48 states. We fit  models both with and without linear time trend, and also calculate our standard errors, as well as OLS, White, cluster and HAC standard errors. We set lag choices $L=3$ for HAC and our standard errors as suggested by \cite{newey1994} ($L=4(T/100)^{2/9}$). The threshold values $M$ chosen by the cross-validation method is $M=0.2$ for the model without state-specific linear trends, and $M=0.1$ with state-specific linear trends. These $M$ values are relatively small, implying the existence of cross-sectional correlations.    The estimated $\beta_{1}, \cdots, \beta_{8}$ with and without linear time trend and their different types of standard errors are presented respectively in Table \ref{div} below. Note that robust standard errors are not necessarily larger than the usual OLS standard errors, as shown in columns corresponding to $se_{CT}$, $se_{DK}$ and $se_{Hard}$.
	
	In Table \ref{div}, OLS estimates with and without  linear time trend are similar to each other. These estimates are also closely comparable to the results obtained in \cite{wolfers2006did}. The OLS estimates indicate that divorce rates rose soon after the law reform. However, within a decade, divorce rates had fallen over time. Most of the coefficient estimates are statistically significant at the $5\%$ level using usual OLS standard errors. According to the cluster standard errors, however, the only significant estimates are 11-15+ after the reform in the model without linear time trend. We use our method of correcting standard error estimates for heteroskedasticity, serial correlation and also cross-sectional correlation. In the model without linear trend, the estimates for 3-4 and 7-15+ are significant. On the other hand, the estimates for 1-4 are significant when linear trend is added.
Our estimated standard errors are close to those  of $se_{CT}$ and $se_{DK}$, which allow arbitrary cross-section correlations. The result indicates non-negligible cross-sectional correlations.
The result is also consistent with \cite{kim2014}, who used the interactive fixed effects approach. The latter approach is suitable for models with strong cross-sectional correlations.

	\section{Conclusions} \label{sec5}
	This paper studies the standard error problem for the OLS estimator in linear panel models, and proposes a new standard-error estimator that is robust to heteroskedasticity, serial and cross-sectional correlations when clusters are unknown. 
	Simulated experiments demonstrate the robustness of the new standard-error estimator to various correlation structures.

	\begin{table}[h!tbp]
		\centering
		\caption{Null rejection probabilities, 5\% level. Two-tailed test of $H_0 : \beta=1$.
			Case 1: No cross-sectional correlation ($\gamma=0$).}\label{no_cross}
		\begin{tabular}{lllrrrrrrrrrr}
			\toprule
			&& & \multicolumn{4}{c}{$\widehat{V}_{\text{Hard}}$} & & $\widehat{V}_{HAC}$ &  $\widehat{V}_{DK}$ & $\widehat{V}_{CX}$ & $\widehat{V}_{CT}$ & $\widehat{V}_{W}$   \\
			\cline{4-7}
			N  &   T   &    L$\backslash$ M   &0.10       & 0.15      & 0.20      &   0.25    &&&&& \\
			\midrule
			&&& \multicolumn{10}{c}{A. No serial correlation: $\rho=0$} \\ \cline{4-13}
			\noalign{\vskip 2mm}
			50	&	100	&	3	&	.067	&	.065	&	.065	&	.067	&&	.057	&	.068	&	.059	&	.058	&	.054	\\
			&		&	7	&	.070	&	.066	&	.070	&	.062	&&	.058	&	.073	&	.059	&	.058	&	.054	\\
			&		&	11	&	.082	&	.071	&	.056	&	.055	&&	.057	&	.088	&	.059	&	.058	&	.054	\\			
			\noalign{\vskip 2mm}
			50	&	200	&	3	&	.054	&	.053	&	.053	&	.051	&&	.046	&	.053	&	.057	&	.044	&	.047	\\
			&		&	7	&	.055	&	.054	&	.056	&	.053	&&	.047	&	.056	&	.057	&	.044	&	.047	\\
			&		&	11	&	.056	&	.054	&	.051	&	.047	&&	.047	&	.061	&	.057	&	.044	&	.047	\\			
			\noalign{\vskip 2mm}
			200	&	100	&	3	&	.062	&	.065	&	.059	&	.060	&&	.047	&	.065	&	.051	&	.055	&	.047	\\
			&		&	7	&	.071	&	.066	&	.057	&	.050	&&	.047	&	.075	&	.051	&	.055	&	.047	\\
			&		&	11	&	.079	&	.065	&	.057	&	.047	&&	.047	&	.091	&	.051	&	.055	&	.047	\\		
			\noalign{\vskip 2mm}
			200	&	200	&	3	&	.051	&	.051	&	.053	&	.052	&&	.048	&	.051	&	.051	&	.051	&	.048	\\
			&		&	7	&	.057	&	.055	&	.055	&	.054	&&	.047	&	.057	&	.051	&	.051	&	.048	\\
			&		&	11	&	.058	&	.052	&	.049	&	.046	&&	.047	&	.060	&	.051	&	.051	&	.048	\\ 
			\midrule
			&&& \multicolumn{10}{c}{B. Serial correlation: $\rho=0.5$} \\ \cline{4-13}
			\noalign{\vskip 2mm}			
			50	&	100	&	3	&	.077	&	.078	&	.081	&	.078	&&	.070	&	.078	&	.065	&	.104	&	.104	\\
			&		&	7	&	.085	&	.084	&	.078	&	.076	&&	.068	&	.083	&	.065	&	.104	&	.104	\\
			&		&	11	&	.086	&	.082	&	.070	&	.066	&&	.067	&	.091	&	.065	&	.104	&	.104	\\		
			\noalign{\vskip 2mm}
			50	&	200	&	3	&	.070	&	.071	&	.075	&	.074	&&	.069	&	.071	&	.067	&	.096	&	.100	\\
			&		&	7	&	.073	&	.070	&	.068	&	.067	&&	.063	&	.072	&	.067	&	.096	&	.100	\\
			&		&	11	&	.077	&	.074	&	.065	&	.064	&&	.061	&	.072	&	.067	&	.096	&	.100	\\		
			\noalign{\vskip 2mm}
			200	&	100	&	3	&	.078	&	.080	&	.080	&	.077	&&	.065	&	.080	&	.053	&	.103	&	.094	\\
			&		&	7	&	.083	&	.078	&	.072	&	.061	&&	.057	&	.082	&	.053	&	.103	&	.094	\\
			&		&	11	&	.087	&	.071	&	.059	&	.058	&&	.055	&	.105	&	.053	&	.103	&	.094	\\
			\noalign{\vskip 2mm}
			200	&	200	&	3	&	.057	&	.056	&	.055	&	.056	&&	.052	&	.057	&	.045	&	.085	&	.082	\\
			&		&	7	&	.059	&	.054	&	.053	&	.051	&&	.048	&	.064	&	.045	&	.085	&	.082	\\
			&		&	11	&	.064	&	.057	&	.054	&	.047	&&	.047	&	.067	&	.045	&	.085	&	.082	\\
			\bottomrule
		\end{tabular}
	\end{table}

	\newpage

	\begin{table}[h!tbp]
		\centering
		\caption{Null rejection probabilities, 5\% level. Two-tailed test of $H_0 : \beta=1$.
			Case 1: Cross-sectional correlation ($\gamma=1$).}\label{crosssectional}
		\begin{tabular}{lllrrrrrrrrrr}
			\toprule
			&& & \multicolumn{4}{c}{$\widehat{V}_{\text{Hard}}$} & & $\widehat{V}_{HAC}$ &  $\widehat{V}_{DK}$ & $\widehat{V}_{CX}$ & $\widehat{V}_{CT}$ & $\widehat{V}_{W}$   \\
			\cline{4-7}
			N  &   T   &    L$\backslash$ M   &0.10       & 0.15      & 0.20      &   0.25    &&&&& \\
			\midrule
			&&& \multicolumn{10}{c}{A. No serial correlation: $\rho=0$} \\ \cline{4-13}	
			\noalign{\vskip 2mm}					
			50	&	100	&	3	&	.054	&	.054	&	.055	&	.055	&&	.142	&	.055	&	.152	&	.054	&	.140	\\
			&		&	7	&	.066	&	.064	&	.063	&	.069	&&	.141	&	.068	&	.152	&	.054	&	.140	\\
			&		&	11	&	.078	&	.077	&	.082	&	.109	&&	.145	&	.079	&	.152	&	.054	&	.140	\\
			\noalign{\vskip 2mm}
			50	&	200	&	3	&	.046	&	.049	&	.046	&	.047	&&	.133	&	.049	&	.145	&	.043	&	.132	\\
			&		&	7	&	.054	&	.055	&	.060	&	.060	&&	.134	&	.052	&	.145	&	.043	&	.132	\\
			&		&	11	&	.059	&	.062	&	.063	&	.073	&&	.135	&	.058	&	.145	&	.043	&	.132	\\
			\noalign{\vskip 2mm}
			200	&	100	&	3	&	.060	&	.060	&	.060	&	.064	&&	.148	&	.058	&	.150	&	.053	&	.148	\\
			&		&	7	&	.069	&	.073	&	.075	&	.080	&&	.149	&	.067	&	.150	&	.053	&	.148	\\
			&		&	11	&	.086	&	.085	&	.096	&	.126	&&	.151	&	.084	&	.150	&	.053	&	.148	\\
			\noalign{\vskip 2mm}
			200	&	200	&	3	&	.050	&	.051	&	.051	&	.050	&&	.121	&	.050	&	.128	&	.049	&	.121	\\
			&		&	7	&	.057	&	.058	&	.057	&	.057	&&	.122	&	.058	&	.128	&	.049	&	.121	\\
			&		&	11	&	.063	&	.062	&	.064	&	.079	&&	.123	&	.062	&	.128	&	.049	&	.121	\\
			\midrule
			&&& \multicolumn{10}{c}{B. Serial correlation: $\rho=0.3$} \\ \cline{4-13}			
			\noalign{\vskip 2mm}			
			50	&	100	&	3	&	.070	&	.069	&	.069	&	.067	&&	.150	&	.069	&	.155	&	.074	&	.176	\\
			&		&	7	&	.074	&	.075	&	.073	&	.078	&&	.150	&	.077	&	.155	&	.074	&	.176	\\
			&		&	11	&	.083	&	.079	&	.093	&	.108	&&	.150	&	.085	&	.155	&	.074	&	.176	\\
			\noalign{\vskip 2mm}
			50	&	200	&	3	&	.058	&	.058	&	.058	&	.058	&&	.150	&	.058	&	.146	&	.069	&	.171	\\
			&		&	7	&	.055	&	.060	&	.062	&	.061	&&	.142	&	.056	&	.146	&	.069	&	.171	\\
			&		&	11	&	.059	&	.063	&	.071	&	.082	&&	.142	&	.060	&	.146	&	.069	&	.171	\\		
			\noalign{\vskip 2mm}
			200	&	100	&	3	&	.078	&	.076	&	.077	&	.072	&&	.162	&	.080	&	.157	&	.091	&	.185	\\
			&		&	7	&	.083	&	.087	&	.083	&	.084	&&	.160	&	.086	&	.157	&	.091	&	.185	\\
			&		&	11	&	.097	&	.089	&	.103	&	.133	&&	.159	&	.101	&	.157	&	.091	&	.185	\\
			\noalign{\vskip 2mm}
			200	&	200	&	3	&	.055	&	.055	&	.054	&	.056	&&	.132	&	.056	&	.133	&	.068	&	.157	\\
			&		&	7	&	.053	&	.051	&	.056	&	.057	&&	.130	&	.057	&	.133	&	.068	&	.157	\\
			&		&	11	&	.057	&	.059	&	.065	&	.078	&&	.130	&	.061	&	.133	&	.068	&	.157	\\			
			\bottomrule
		\end{tabular}
	\end{table}

	\begin{table}[h!tbp]
		\centering
		\caption{Null rejection probabilities, 5\% level. Two-tailed test of $H_0 : \beta=1$.
			Case 1: Both strong serial and cross-sectional correlations ($\rho=0.9, \gamma=1$).}\label{bothcorr3}
		\begin{tabular}{lllrrrrrrrrrr}
			\toprule
			&& & \multicolumn{4}{c}{$\widehat{V}_{\text{Hard}}$} & & $\widehat{V}_{HAC}$ &  $\widehat{V}_{DK}$ & $\widehat{V}_{CX}$ & $\widehat{V}_{CT}$ & $\widehat{V}_{W}$   \\
			\cline{4-7}
			N  &   T   &    L$\backslash$ M   &0.10       & 0.15      & 0.20      &   0.25    &&&&& &\\
			\midrule
			50	&	100	&	3	&	.096	&	.097	&	.098	&	.098	&&	.180	&	.098	&	.168	&	.145	&	.271	\\
			&		&	7	&	.101	&	.101	&	.100	&	.101	&&	.173	&	.102	&	.168	&	.145	&	.271	\\
			&		&	11	&	.113	&	.101	&	.110	&	.128	&&	.174	&	.113	&	.168	&	.145	&	.271	\\
			\noalign{\vskip 2mm}
			50	&	200	&	3	&	.091	&	.092	&	.092	&	.092	&&	.194	&	.092	&	.180	&	.157	&	.286	\\
			&		&	7	&	.088	&	.087	&	.090	&	.093	&&	.185	&	.087	&	.180	&	.157	&	.286	\\
			&		&	11	&	.092	&	.092	&	.099	&	.114	&&	.182	&	.092	&	.180	&	.157	&	.286	\\
			\noalign{\vskip 2mm}
			200	&	100	&	3	&	.089	&	.087	&	.087	&	.086	&&	.193	&	.089	&	.146	&	.144	&	.256	\\
			&		&	7	&	.092	&	.089	&	.095	&	.097	&&	.178	&	.097	&	.146	&	.144	&	.256	\\
			&		&	11	&	.100	&	.096	&	.109	&	.128	&&	.173	&	.109	&	.146	&	.144	&	.256	\\
			\noalign{\vskip 2mm}
			200	&	200	&	3	&	.069	&	.069	&	.069	&	.067	&&	.146	&	.068	&	.125	&	.121	&	.226	\\
			&		&	7	&	.066	&	.068	&	.069	&	.067	&&	.136	&	.069	&	.125	&	.121	&	.226	\\
			&		&	11	&	.072	&	.071	&	.076	&	.087	&&	.133	&	.072	&	.125	&	.121	&	.226	\\
			\bottomrule
		\end{tabular}
	\end{table}

	\begin{table}[h!tbp]
		\centering
		\caption{Null rejection probabilities, 5\% level. Two-tailed test of $H_0 : \beta=1$.
			Case 2: Errors with Spatial AR(1) structure  ($ \psi=0.5$).}\label{SAR}
		\begin{tabular}{lllrrrrrrrrrr}
			\toprule
			&& & \multicolumn{4}{c}{$\widehat{V}_{\text{Hard}}$} & & $\widehat{V}_{HAC}$ &  $\widehat{V}_{DK}$ & $\widehat{V}_{CX}$ & $\widehat{V}_{CT}$ & $\widehat{V}_{W}$   \\
			\cline{4-7}
			N  &   T   &    L$\backslash$ M   &0.10       & 0.15      & 0.20      &   0.25    &&&&& &\\
			\midrule
			50	&	100	&	3	&	.061	&	.060	&	.062	&	.057	&&	.124	&	.059	&	.143	&	.053	&	.125	\\
				&		&	7	&	.068	&	.068	&	.073	&	.077	&&	.125	&	.067	&	.143	&	.053	&	.125	\\
				&		&	11	&	.086	&	.083	&	.089	&	.110	&&	.128	&	.085	&	.143	&	.053	&	.125	\\
			\noalign{\vskip 2mm}
			50	&	200	&	3	&	.046	&	.046	&	.047	&	.047	&&	.113	&	.046	&	.130	&	.043	&	.111	\\
				&		&	7	&	.052	&	.053	&	.053	&	.059	&&	.114	&	.048	&	.130	&	.043	&	.111	\\
				&		&	11	&	.052	&	.059	&	.065	&	.083	&&	.112	&	.061	&	.130	&	.043	&	.111	\\			
			\noalign{\vskip 2mm}
			200	&	100	&	3	&	.061	&	.057	&	.058	&	.057	&&	.123	&	.062	&	.120	&	.051	&	.122	\\
				&		&	7	&	.070	&	.071	&	.070	&	.081	&&	.122	&	.068	&	.120	&	.051	&	.122	\\
				&		&	11	&	.082	&	.080	&	.101	&	.117	&&	.121	&	.088	&	.120	&	.051	&	.122	\\
			\noalign{\vskip 2mm}
			200	&	200	&	3	&	.055	&	.055	&	.053	&	.050	&&	.124	&	.056	&	.125	&	.049	&	.123	\\
				&		&	7	&	.064	&	.061	&	.062	&	.064	&&	.123	&	.064	&	.125	&	.049	&	.123	\\
				&		&	11	&	.069	&	.070	&	.083	&	.099	&&	.123	&	.068	&	.125	&	.049	&	.123	\\
			\bottomrule
		\end{tabular}
	\end{table}

	\begin{table}[h!tbp]
		\centering
		\caption{Null rejection probabilities, 5\% level. Two-tailed test of $H_0 : \beta=1$.
			Case 3: Errors with Factor structure ($\rho_{F}=0.9, \rho_{\lambda}=0.3$).}\label{error factor structure}
		\begin{tabular}{lllrrrrrrrrrr}
			\toprule
			&& & \multicolumn{4}{c}{$\widehat{V}_{\text{Hard}}$} & & $\widehat{V}_{HAC}$ &  $\widehat{V}_{DK}$ & $\widehat{V}_{CX}$ & $\widehat{V}_{CT}$ & $\widehat{V}_{W}$   \\
			\cline{4-7}
			N  &   T   &    L$\backslash$ M   &0.10       & 0.15      & 0.20      &   0.25    &&&&& &\\
			\midrule
			50	&	100	&	3	&	.081	&	.080	&	.078	&	.081	&&	.129	&	.078	&	.102	&	.130	&	.202	\\
			&		&	7	&	.091	&	.084	&	.079	&	.082	&&	.115	&	.093	&	.102	&	.130	&	.202	\\
			&		&	11	&	.103	&	.086	&	.094	&	.086	&&	.115	&	.108	&	.102	&	.130	&	.202	\\
			\noalign{\vskip 2mm}
			50	&	200	&	3	&	.083	&	.082	&	.081	&	.081	&&	.117	&	.080	&	.095	&	.132	&	.183	\\
			&		&	7	&	.066	&	.066	&	.065	&	.073	&&	.107	&	.065	&	.095	&	.132	&	.183	\\
			&		&	11	&	.072	&	.076	&	.080	&	.084	&&	.104	&	.072	&	.095	&	.132	&	.183	\\			
			\noalign{\vskip 2mm}
			200	&	100	&	3	&	.076	&	.074	&	.076	&	.074	&&	.109	&	.076	&	.086	&	.121	&	.167	\\
			&		&	7	&	.077	&	.080	&	.077	&	.074	&&	.105	&	.083	&	.086	&	.121	&	.167	\\
			&		&	11	&	.081	&	.076	&	.084	&	.094	&&	.103	&	.096	&	.086	&	.121	&	.167	\\			
			\noalign{\vskip 2mm}
			200	&	200	&	3	&	.072	&	.072	&	.073	&	.069	&&	.115	&	.071	&	.090	&	.126	&	.184	\\
			&		&	7	&	.070	&	.067	&	.071	&	.074	&&	.106	&	.068	&	.090	&	.126	&	.184	\\
			&		&	11	&	.074	&	.073	&	.073	&	.075	&&	.104	&	.072	&	.090	&	.126	&	.184	\\		
			\bottomrule
		\end{tabular}
	\end{table}

	\begin{table}[h!tbp]
		\centering	
		\caption{Empirical application: effects of divorce law refrom with state and year fixed effects: US state level data annual from 1956 to 1988, dependent variable is divorce rate per 1000 persons per year. OLS estimates and standard errors (using state population weights).}\label{div}
		\begin{tabular}{p{2.3cm}p{.9cm}p{.9cm}p{.9cm}p{.9cm}p{.9cm}p{.9cm}p{.9cm}p{.9cm}}
			\toprule
			Effects: & $\hat{\beta}_{OLS}$ &$ se_{OLS}$ & $se_{W}$ & $se_{CX}$ & $se_{CT}$ & $se_{HAC}$& $se_{DK}$ &  $se_{Hard} $  \\
			\midrule
			\multicolumn{9}{c}{Panel A: Without state-specific linear time trends} \\																						
			\noalign{\vskip 2mm}																						
			1\textendash2 years	&	.256	&	.086*	&	.140		&	.189		&	.139		&	.172	&	.155	&	.148	\\		
			﻿3\textendash4 years	&	.209	&	.086*	&	.081*	&	.159		&	.075*	&	.114	&	.104*	&	.089*	\\				
			﻿5\textendash6 years	&	.126	&	.086		&	.073		&	.168		&	.064*	&	.105	&	.088	&	.069	\\	
			﻿7\textendash8 years	&	.105	&	.086		&	.070		&	.165		&	.059		&	.100	&	.065	&	.040*	\\
			﻿9\textendash10 years	&	-.122	&	.085		&	.060*	&	.161		&	.041*	&	.088	&	.058*	&	.054*	\\		
			﻿11\textendash12 years	&	-.344	&	.085*	&	.071*	&	.173*	&	.043*	&	.101*	&	.056*	&	.075*	\\				
			﻿13\textendash14 years	&	-.496	&	.085*	&	.074*	&	.188*	&	.050*	&	.110*	&	.054*	&	.062*	\\				
			﻿15+ years	&	-.508	&	.081*	&	.089*	&	.223*	&	.048*	&	.139*	&	.061*	&	.077*	\\				
			\midrule
			\multicolumn{9}{c}{Panel B: With state-specific linear time trends} \\																						
			\noalign{\vskip 2mm}																						
			1\textendash2 years	&	.286	&	.064*		&	.152		&	.206		&	.143*	&	.185	&	.145*	&	.140*	\\	
			﻿3\textendash4 years	&	.254	&	.071*		&	.099*		&	.171		&	.102*	&	.140	&	.134	&	.126*	\\	
			﻿5\textendash6 years	&	.186	&	.079*		&	.102		&	.206		&	.110		&	.145	&	.148	&	.143	\\
			﻿7\textendash8 years&	.177	&	.086*		&	.109		&	.230		&	.120		&	.153	&	.155	&	.146	\\
			﻿9\textendash10 years	&	-.037	&	.093		&	.111		&	.241		&	.120		&	.156	&	.164	&	.154	\\
			﻿11\textendash12 years	&	-.247	&	.100*		&	.128		&	.268		&	.141		& .179	&	.196	&	.183	\\	
			﻿13\textendash14 years	&	-.386	&	.108*		&	.137*		&	.296		&	.164*	& .193*	&	.218	&	.209	\\		
			﻿15+ years	&	-.414	&	.120*		&	.158*		&	.337		&	.186*	&	.221	&	.251	&	.243	\\	
			\bottomrule
		\end{tabular}\\
		\begin{tablenotes}
		\item \textbf{Note:} Standard errors with asterisks indicate significance at 5\% level using $N(0,1)$ critical values; $se_{OLS}$ and $se_{W}$ refer to OLS and White standard errors respectively; $se_{CX}$ and $se_{CT}$  are clustered standard errors suggested by \cite{arellano1987}; $se_{HAC}$ and $se_{DK}$ are two types of Newey-West HAC estimator as explained in the text; $se_{Hard}$ is our standard error. Bartlett kernel with lag length $L=3$ is used for $se_{HAC}$, $se_{DK}$ and $se_{Hard}$. The threshold value for $se_{Hard}$ by the cross-validation is  $M=0.2$ (for the first panel)  and $M=0.1$ (for the second panel). 
	\end{tablenotes}
			\end{table}

	\clearpage
	\newpage

	\appendix
	\section{Appendix} \label{app}
	
	Throughout the proof, $max_{i}$, $max_{t}$, $max_{h}$, $max_{ij}$, $max_{it}$, $\sum_{i}$, $\sum_{t}$, and $\sum_{ij}$ denote $max_{i\leq N}$, $max_{t\leq T}$, $max_{h\leq L}$, $max_{i,j}$, $max_{i, t}$, $\sum_{i=1}^{N}$, $\sum_{t=1}^{T}$, and $\sum_{i=1}^{N}\sum_{j=1}^{N}$ respectively.\\
	\subsection{Proof of Theorem 2.1}
	First let
	\begin{equation*}
	V_{L} = \cfrac{1}{NT}\sum\limits_{t}Ex_{t}'u_{t}u_{t}'x_{t} +  \cfrac{1}{NT}\sum\limits_{h=1}^{L}\omega(h,L)\sum\limits_{t=h+1}^{T}[Ex_{t}'u_{t}u_{t-h}'x_{t-h}+Ex_{t-h}'u_{t-h}u_{t}'x_{t}].
	\end{equation*}
	We need following lemmas to prove the main results. \\
	\begin{lem} \label{lem1}
	(i) $ \|V - V_{L}\| \leq C\sum\limits_{h=L}^{T-1}\alpha_{NT}(h) + C\sum\limits_{h=1}^{L}(1-\omega(h,L))\alpha_{NT}(h)$.\\	
	(ii) $\max_i|V_{u,ii} - \var(\frac{1}{\sqrt{T}}\sum_{t=1}^Tx_{it}u_{it})|=o(1)$.\\
	(iii) $\min_i\lambda_{\min}(V_{u,ii})>c$.
	\end{lem}
	\begin{proof}
	(i)	First note that
		\begin{align*}
		\|Ex_{t}'u_{t}u_{t-h}'x_{t-h}+Ex_{t-h}'u_{t-h}u_{t}'x_{t}\| &\leq E\|x_{t}\|\|E(u_{t}u_{t-h}'|X)\|\|x_{t-h}\|+E\|x_{t-h}\|\|E(u_{t-h}u_{t}'|X)\|\|x_{t}\|\\
		& \leq \alpha_{NT}(h)E\|x_{t}\|\|x_{t-h}\|  \leq CN\alpha_{NT}(h).
		\end{align*}
		Hence for some $C,c>0,$
		\begin{align*}
		\|V - V_{L}\| & \leq \|\cfrac{1}{NT}\sum\limits_{h=L+1}^{T-1}\sum\limits_{t=h+1}^{T}[Ex_{t}'u_{t}u_{t-h}'x_{t-h}+Ex_{t-h}'u_{t-h}u_{t}'x_{t}]\| \\
		&  \;\; + \|\cfrac{1}{NT}\sum\limits_{h=1}^{L}(1-\omega(h,L))\sum\limits_{t=h+1}^{T}[Ex_{t}'u_{t}u_{t-h}'x_{t-h} + Ex_{t-h}'u_{t-h}u_{t}'x_{t}]\|\\
		& \leq C\frac{1}{T}\sum\limits_{h=L+1}^{T-1}\sum\limits_{t=h+1}^{T}\alpha_{NT}(h^{c}) + C\frac{1}{T}\sum\limits_{h=1}^{L}(1-\omega(h,L))\sum\limits_{t=h+1}^{T}\alpha_{NT}(h^{c}) \\
		&\leq C\sum\limits_{h>L}\alpha_{NT}(h^{c}) + C\sum\limits_{h=1}^{L} |1-\omega(h,L)|\alpha_{NT}(h^{c}) = o (1).
		\end{align*}
		The second term of the last equation goes to zero due to Assumption \ref{assumption2}(iii) and the dominated convergence theorem, noting that $|1-\omega(h,L)|\alpha_{NT}(h^{c})\leq C\alpha_{NT}(h^{c})$ and $\alpha_{NT}(h^{c})$ is summable over $h$.
		
	(ii)	The proof for $\max_i|V_{u,ii} - \var(\frac{1}{\sqrt{T}}\sum_{t=1}^Tx_{it}u_{it})|=o(1)$ follows from the same argument.
		
		(iii) The result follows from (ii) and the assumption that $\min_i\lambda_{\min}( \var(\frac{1}{\sqrt{T}}\sum_{t=1}^Tx_{it}u_{it}))>c$.
		
	\end{proof}
	
	\begin{lem} \label{lem2}
		Suppose $\log N = o(T).$ For $f(t,h,L) = \omega(h,L)1\{t>h\}$,
		\begin{equation*}
		\max_{h}\max_{i,j}\|\frac{1}{T}\sum_{t=1}^{T}x_{it}u_{it}u_{j,t-h}x_{j,t-h}'f(t,h,L)-Ex_{it}u_{it}u_{j,t-h}x_{j,t-h}'f(t,h,L)\| = O_{P}(\sqrt{\frac{\log(LN)}{T}}).
		\end{equation*}
	\end{lem}
	\begin{proof}
		The left hand side can be written as\\~
		$\max_{h}\max_{ij}\|\frac{1}{T}\sum_{t}Z_{h,ij,t}\|$, where $Z_{h,ij,t} = f(t,h,L)(x_{it}\varepsilon_{it}\varepsilon_{j,t-h}x_{j,t-h}' - Ex_{it}\varepsilon_{it}\varepsilon_{j,t-h}x_{j,t-h}' ). $\\~
		For convenience, assume that $\dim(Z_{h,ij,t}) = 1$ and there is no serial correlation.  Set $\alpha_{n} = \sqrt{\frac{\log(LN)}{T}}$ and $c^2=2C$ for $c, C>0$. Then, by using Bernstein Inequality and exponential tail conditions, and that $f(t,h,L)$ is bounded,
		\begin{align*}
		P(\max\limits_{h \leq L}\max\limits_{ij}|\frac{1}{T}\sum_{t=1}^{T}Z_{h,ij,t}| > c\alpha_{n}) &\leq LN^2\max\limits_{h \leq L}\max\limits_{ij}P(|\cfrac{1}{T}\sum\limits_{t=1}^{T}Z_{h,ij,t}| > c\alpha_{n})\\
		& \leq LN^2exp(-\cfrac{Tc^2\alpha_{n}^2}{C})\\
		& \leq exp(\log(LN)-\cfrac{Tc^2\alpha_{n}^2}{C})\\
		& = exp(-\log(LN))\\
		&= \cfrac{1}{LN} \rightarrow 0. \;\;
		\end{align*}
	\end{proof}
	
	\begin{lem} \label{lem3}
		Suppose $\log N = o(T).$ For $f(t,h,L) = \omega(h,L)1\{t>h\}$,
		\begin{equation*}
		 \max_{h}\max_{i,j}\|\frac{1}{T}\sum_{t=1}^{T}x_{it}\widehat{u}_{it}\widehat{u}_{j,t-h}x_{j,t-h}'f(t,h,L)-x_{it}u_{it}u_{j,t-h}x_{j,t-h}'f(t,h,L)\| =O_{P}(\frac{1}{T}\sqrt{\frac{\log(LN)}{N}}).
		\end{equation*}
	\end{lem}
	\begin{proof}
		The left hand side is bounded by $a_1+a_2+a_3$, where
		\begin{align*}
		&a_1 = \max_{h}\max_{i,j}\|\frac{1}{T}\sum_{t=1}^{T}x_{it}(\widehat{u}_{it}-u_{it})(\widehat{u}_{j,t-h}-u_{j,t-h})x_{j,t-h}'f(t,h,L)\|\\
		&a_2 = \max_{h}\max_{i,j}\|\frac{1}{T}\sum_{t=1}^{T}x_{it}u_{it}(\widehat{u}_{j,t-h}-u_{j,t-h})x_{j,t-h}'f(t,h,L)\| \\
		&a_3 = \max_{h}\max_{i,j}\|\frac{1}{T}\sum_{t=1}^{T}x_{it}(\widehat{u}_{it}-u_{it})u_{j,t-h}x_{j,t-h}'f(t,h,L)\|.
		\end{align*}
		For simplicity, let's assume $\dim(x_{it}) =1$. Then
		\begin{align*}
		a_1 &\leq \|\widehat{\beta}-\beta\|^2\max_{h}\max_{i,j}\|\frac{1}{T}\sum_{t=1}^{T}x_{it}x_{it}x_{j,t-h}x_{j,t-h}f(t,h,L)\|\\
		&\leq O_{P}(\frac{1}{NT})\max_{ij}\frac{1}{T}\sum_{t=1}^{T}\|x_{it}\|^4 = O_{P}(\frac{1}{NT}).
		\end{align*}
		By using Bernstein Inequality for weakiy dependent data and exponential tail conditions, and that $f(t,h,L)$ is bounded,
		\begin{align*}
		a_2 &\leq \|\widehat{\beta}-\beta\|\max_{h}\max_{i,j}\|\frac{1}{T}\sum_{t=1}^{T}x_{it}u_{it}x_{j,t-h}x_{j,t-h}f(t,h,L)\| \\
		&\leq O_{P}(\frac{1}{\sqrt{NT}})O_{P}(\sqrt{\frac{\log(LN)}{T}}) \\
		&= O_{P}(\frac{1}{T}\sqrt{\frac{\log(LN)}{N}}).
		\end{align*}
		$a_{3}$ is bounded using the same argument. Together,
		\begin{equation*}
		 \max_{h}\max_{i,j}\|\frac{1}{T}\sum_{t=1}^{T}x_{it}\widehat{u}_{it}\widehat{u}_{j,t-h}x_{j,t-h}'f(t,h,L)-x_{it}u_{it}u_{j,t-h}x_{j,t-h}'f(t,h,L)\| =O_{P}(\frac{1}{T}\sqrt{\frac{\log(LN)}{N}}).
		\end{equation*}
	\end{proof}
	\noindent
	\begin{proof}[\textbf{Proof of Theorem \ref{asdistribution}.} ]
		It suffice to prove $\|\widehat{V}-V\| = o_{P}(1)$. By Lemma \ref{lem1}, we have
		\begin{align*}
		\|\widehat{V}-V\|  \leq \|\widehat{V}-V_{L}\| + C\sum\limits_{h>L}\alpha_{NT}(h)+C\sum\limits_{h=1}^{L}(1-\omega(h,L))\alpha_{NT}(h).
		\end{align*}
		The remaining proof is that of $\|\widehat{V}-V_{L}\| = o_{P}(1)$, given below. \\~
		\\~
		\textbf{Main proof of the convergence of $\|\widehat{V}-V_{L}\|$}\\~
		Note that $V_{L} = \frac{1}{N}\sum_{ij}V_{u,ij}$, $\widehat{V} = \frac{1}{N}\sum_{ij}\widehat{S}_{u,ij}$. Hence
		\begin{align*}
		\|\widehat{V}-V_{L}\| \leq \cfrac{1}{N}\sum\limits_{\widehat{S}_{u,ij}=0}\|V_{u,ij}-\widehat{S}_{u,ij}\| + \cfrac{1}{N}\sum\limits_{\widehat{S}_{u,ij} \neq 0}\|V_{u,ij}-\widehat{S}_{u,ij}\|.
		\end{align*}
		Note that $\|S_{u,ij}-V_{u,ij}\| < \frac{1}{2}\lambda_{ij}$ for $\forall (i, j)$ and $C_1>0$
		\begin{align*}
		\|S_{u,ii}\| &\geq \|V_{u,ii}\| - \|S_{u,ii}-V_{u,ii}\|\\
		& \geq \|V_{u,ii}\| - \max\limits_{ij}\|S_{u,ii}-V_{u,ii}\|\\
		& \geq \|V_{u,ii}\| - C\omega_{NT} > C_1.
		\end{align*}
		From Assumption \ref{assumption3}, $\|V_{u,ii}\| > c_{1} > 0$, then, $\lambda_{ij} = M\omega_{NT}\sqrt{\|S_{u,ii}\|\|S_{u,jj}\|} > c_{1}M\omega_{NT} > 2c_{1}\omega_{NT}$. Then, $\frac{\lambda_{ij}}{2} > c\omega_{NT} \geq \max\limits_{ij}\|S_{u,ij}-V_{u,ij}\|$.  Therefore, $\|S_{u,ij}-V_{u,ij}\| < \frac{1}{2}\lambda_{ij}$ for $\forall (i, j)$\\~
		Recall $\rho_{ij,h} = \sup_{X}\max_{t}|E(u_{it}u_{j,t-h}|X)| + |E(u_{i,t-h}u_{jt}|X)|$. Then,
		\begin{align*}
		\|V_{u,ij}\| & \leq \|\cfrac{1}{T}\sum\limits_{t}Ex_{it}u_{it}u_{jt}x_{jt}'+\cfrac{1}{T}\sum\limits_{h=1}^{L}\omega(h,L)\sum\limits_{t=h+1}^{T}[Ex_{it}u_{it}u_{j,t-h}x_{j,t-h}' + Ex_{i,t-h}u_{i,t-h}u_{jt}x_{jt}']\|\\
		&\leq C\rho_{ij,0}/2 + C\cfrac{1}{T}\sum\limits_{h=1}^{L}\omega(h,L)\sum\limits_{t=h+1}^{T}\rho_{ij,h} \leq C\sum\limits_{h=0}^{L}\rho_{ij,h}.
		\end{align*}
		Hence, on the event $\max_{ij}\|S_{u,ij}-V_{u,ij}\| \leq C\omega_{NT}$,
		\begin{align*}
		\cfrac{1}{N}\sum\limits_{\widehat{S}_{u,ij}=0}\|V_{u,ij}-\widehat{S}_{u,ij}\| & \leq \cfrac{1}{N}\sum\limits_{\widehat{S}_{u,ij}=0}\|V_{u,ij}\| \leq \cfrac{1}{N}\sum\limits_{ij}\|V_{u,ij}\|1\{\|S_{u,ij}\|< \lambda_{ij}\}\\
		&= \cfrac{1}{N}\sum\limits_{ij}\|V_{u,ij}\|1\{\|V_{u,ij}\|< \|S_{u,ij}\| + \|S_{u,ij}-V_{u,ij}\|, \|S_{u,ij}\|< \lambda_{ij}\}\\
		&\leq \cfrac{1}{N}\sum\limits_{ij}\|V_{u,ij}\|\cfrac{(1.5\lambda_{ij})^{1-q}}{\|V_{u,ij}\|^{1-q}}1\{\|V_{u,ij}\|< 1.5\lambda_{ij}\}\\
		&\leq \cfrac{1}{N}\sum\limits_{ij}\|V_{u,ij}\|^{q}(1.5\lambda_{ij})^{1-q} \leq C\omega_{NT}^{1-q}\cfrac{1}{N}\sum\limits_{ij}\|V_{u,ij}\|^{q}\\
		&\leq C\omega_{NT}^{1-q}\max\limits_{i}\sum\limits_{j}(\sum\limits_{h=0}^{L}\rho_{ij,h})^{q}.
		\end{align*}
		On the other hand, on the event $\max_{ij}\|S_{u,ij}-V_{u,ij}\| \leq C\omega_{NT}$,
		\begin{align*}
		\cfrac{1}{N}\sum\limits_{\widehat{S}_{u,ij} \neq 0}\|V_{u,ij}-\widehat{S}_{u,ij}\| & \leq \cfrac{1}{N}\sum\limits_{\widehat{S}_{u,ij} \neq 0}\|V_{u,ij}-S_{u,ij}\| + \cfrac{1}{N}\sum\limits_{\widehat{S}_{u,ij} \neq 0}\|S_{u,ij}-\widehat{S}_{u,ij}\| \\
		&\leq \cfrac{1}{N}\sum\limits_{\widehat{S}_{u,ij} \neq 0}0.5\lambda_{ij} + \cfrac{1}{N}\sum\limits_{\widehat{S}_{u,ij} \neq 0}\lambda_{ij} \leq \cfrac{1}{N}\sum\limits_{ij}1.5\lambda_{ij}1\{\|S_{u,ij}\|> \lambda_{ij}\}\\
		&= \cfrac{1}{N}\sum\limits_{ij}1.5\lambda_{ij}1\{\|V_{u,ij}\|> \|S_{u,ij}\| - \|S_{u,ij}-V_{u,ij}\|,\|S_{u,ij}\|> \lambda_{ij}\}\\
		&\leq \cfrac{1}{N}\sum\limits_{ij}1.5\lambda_{ij}\cfrac{\|V_{u,ij}\|^{q}}{(0.5\lambda_{ij})^{q}}1\{\|V_{u,ij}\|>0.5\lambda_{ij}\}\\
		&\leq \cfrac{1}{N}\sum\limits_{ij} C\lambda_{ij}^{1-q}\|V_{u,ij}\|^{q} \leq \cfrac{1}{N}\sum\limits_{ij}\|V_{u,ij}\|^{q}C\omega_{NT}^{1-q} \\
		&\leq C\omega_{NT}^{1-q}\max\limits_{i}\sum\limits_{j}(\sum\limits_{h=0}^{L}\rho_{ij,h})^{q}.
		\end{align*}
		Hence $\|\widehat{V}-V_{L}\| \leq C\omega_{NT}^{1-q}\max_{i}\sum_{j}(\sum_{h=0}^{L}\rho_{ij,h})^{q}$. Therefore, we have
		\begin{align*}\|\widehat{V}-V\| \leq O_{P}(\omega_{NT}^{1-q}\max\limits_{i}\sum\limits_{j}(\sum\limits_{h=0}^{L}\rho_{ij,h})^{q}) + C\sum\limits_{h=L}^{T-1}\alpha_{NT}(h) + C\sum\limits_{h=1}^{L}(1-\omega(h,L))\alpha_{NT}(h). 
		\end{align*}
	\end{proof}
	\begin{proof}[\textbf{Remaining proofs :} $\max_{ij}\|S_{u,ij}-V_{u,ij}\| = O_{P}(\omega_{NT})$, where $\omega_{NT}=L\sqrt{\frac{\log(LN)}{T}}$]
		Recall\\~
		$S_{u,ij} \equiv \cfrac{1}{T}\sum\limits_{t}x_{it}\widehat{u}_{it}\widehat{u}_{jt}x_{jt}'+\cfrac{1}{T}\sum\limits_{h=1}^{L}\omega(h,L)\sum\limits_{t=h+1}^{T}[x_{it}\widehat{u}_{it}\widehat{u}_{j,t-h}x_{j,t-h}' + x_{i,t-h}\widehat{u}_{i,t-h}\widehat{u}_{jt}x_{jt}'],$\\
		$V_{u,ij} \equiv \cfrac{1}{T}\sum\limits_{t}Ex_{it}u_{it}u_{jt}x_{jt}'+\cfrac{1}{T}\sum\limits_{h=1}^{L}\omega(h,L)\sum\limits_{t=h+1}^{T}[Ex_{it}u_{it}u_{j,t-h}x_{j,t-h}' + Ex_{i,t-h}u_{i,t-h}u_{jt}x_{jt}']$.\\~
		Let\\~
		$M_{u,ij} \equiv \cfrac{1}{T}\sum\limits_{t}x_{it}u_{it}u_{jt}x_{jt}'+\cfrac{1}{T}\sum\limits_{h=1}^{L}\omega(h,L)\sum\limits_{t=h+1}^{T}[x_{it}u_{it}u_{j,t-h}x_{j,t-h}' + x_{i,t-h}u_{i,t-h}u_{jt}x_{jt}'].$\\~
		\\~
		We first bound $\max_{ij}\|M_{u,ij}-V_{u,ij}\|$, then bound $\max_{ij}\|S_{u,ij}-M_{u,ij}\| $.\\~
		\\~
		\textbf{Proof of $\max_{ij}\|M_{u,ij}-V_{u,ij}\| = O_{P}(L\sqrt{\frac{\log(LN)}{T}})$}\\~
		Given Lemma \ref{lem2}, we have
		\begin{align*}
		\max_{ij}\|M_{u,ij}-V_{u,ij}\| & \leq O_{P}(\sqrt{\frac{\log (LN)}{T}})\\~
		& +2\max_{ij}\|\frac{1}{T}\sum_{h=1}^{L}\sum_{t=1}^{T}[x_{it}u_{it}u_{j,t-h}x_{j,t-h}'f(t,h,L)-Ex_{it}u_{it}u_{j,t-h}x_{j,t-h}'f(t,h,L)]\|\\~
		&\leq O_{P}(\sqrt{\frac{\log (LN)}{T}}) + LO_{P}(\sqrt{\frac{\log(LN)}{T}}) = O_{P}(L\sqrt{\frac{\log(LN)}{T}}). 
		\end{align*}
		\\~
		\textbf{Prove of $\max_{ij}\|M_{u,ij}-S_{u,ij}\| = O_{P}(\frac{L}{T}\sqrt{\frac{\log(LN)}{N}})$}\\~
		Given Lemma \ref{lem3}, we have
		\begin{align*}
		\max_{ij}\|M_{u,ij}-S_{u,ij}\| & \leq O_{P}(\frac{1}{T}\sqrt{\frac{\log(LN)}{N}}) \\~
		& +2\max_{ij}\|\frac{1}{T}\sum_{h=1}^{L}\sum_{t=1}^{T}[x_{it}\widehat{u}_{it}\widehat{u}_{j,t-h}x_{j,t-h}'f(t,h,L)-x_{it}u_{it}u_{j,t-h}x_{j,t-h}'f(t,h,L))]\|\\~
		&\leq O_{P}(\frac{1}{T}\sqrt{\frac{\log(LN)}{N}})  + LO_{P}(\frac{1}{T}\sqrt{\frac{\log(LN)}{N}}) = O_{P}(\frac{L}{T}\sqrt{\frac{\log(LN)}{N}}).
		\end{align*}
		Together,
		\begin{align*}
		\max_{ij}\|V_{u,ij}-S_{u,ij}\| = O_{P}(L\sqrt{\frac{\log(LN)}{T}})+O_{P}(\frac{L}{T}\sqrt{\frac{\log(LN)}{N}}) = O_{P}(L\sqrt{\frac{\log(LN)}{T}}).
		\end{align*}
	\end{proof}

	\newpage
	\bibliographystyle{economet}
	\bibliography{reference}

\end{document}